\documentclass[11pt]{article}
\usepackage[margin=1in]{geometry}
\usepackage{array,xspace,multirow,hhline,tikz,colortbl,tabularx,booktabs,fixltx2e,amsmath,amsfonts,amsthm}
\usepackage{soul}
\usepackage{pifont}
\usepackage{times}
\usepackage{natbib}
\usepackage{enumerate}
\usepackage{hyperref}

\usepackage{latexsym}
\usepackage{bm}
\usepackage{float}
\usepackage{epsfig}
\usepackage{xspace}
\usepackage{paralist}
\usepackage{enumerate}
\usepackage{cases}
\usepackage{caption}
\usepackage{mathtools}
\usepackage{algpseudocode}
\usepackage{graphicx}
\usepackage{subcaption}

\usepackage{xcolor}
\usepackage{complexity}
\usepackage{multicol}
\usepackage{thm-restate}

\usepackage{multirow}
\usepackage{diagbox}
\usepackage{verbatim}
\usepackage{enumitem}
\usepackage{tikz}
\usepackage{tikz-network}

\usepackage{mfirstuc}

\newtheorem{theorem}{Theorem}[section]

\newtheorem{definition}[theorem]{Definition}

\newtheorem{lemma}[theorem]{Lemma}

\newtheorem{proposition}[theorem]{Proposition}

\newcommand{\PROP}{\mathsf{PROP}}
\newcommand{\MMS}{\mathsf{MMS}}

\usepackage{url}
\newcommand{\ssc}[1]{\textsc{\capitalisewords{\MakeLowercase{#1}}}}
\renewcommand{\epsilon}{\varepsilon}

\begin{document}

\title{Fair Orientations: Proportionality and Equitability}
\author{Ankang Sun\thanks{School of Computer Science and Technology, Shandong University, Qingdao, China. ankang.sun@sdu.edu.cn} \and Ruijie Wang\thanks{Department of Computing, The Hong Kong Polytechnic University, Hong Kong, China. ruijie.wang@connect.polyu.hk} \and Bo Li\thanks{Department of Computing, The Hong Kong Polytechnic University, Hong Kong, China. comp-bo.li@polyu.edu.hk}}
\date{}

\maketitle

\begin{abstract}
We study the fair allocation of indivisible items under relevance constraints, where each agent has a set of relevant items and can only receive items that are relevant to them. 
While the relevance constraint has been studied in recent years, existing work has largely focused on envy-freeness. 
Our work extends this study to other key fairness criteria --- such as proportionality, equitability, and their relaxations --- in settings where the items may be goods, chores, or a mixture of both.
We complement the literature by presenting a picture of the existence and computational complexity of the considered criteria.

\end{abstract}

\section{Introduction}
Fair division is a fundamental and challenging problem in both AI and economics. It involves allocating resources or tasks among participants while ensuring that participants are treated fairly. 
There is no single fairness criterion universally applicable. Among the various proposed fairness notions, envy-freeness \citep{VARIAN197463} and proportionality \citep{steinhaus1949division} are two of the most prominent and well-studied.
Envy-freeness requires that no participant prefers the allocation of another over their own, while proportionality requires that each participant receives a value at least their proportional share, a threshold determined by the total number of participants and agents' valuations.

In the classical setting, any item can be allocated to any participant. However, in practical allocation problems, there are often relevance relationships between resources and participants, which restrict an item from being assigned only to a subset of agents. 
Relevance captures real-world scenarios where certain items are only accessible to specific agents due to geographic restrictions or demand requirements. For example, in allocating railway maintenance jobs between cities A and B, the jobs are typically managed by the local authorities of each city. Similarly, in sports leagues such as the NBA and NFL, matches are held on the home court of one of the competing teams. Motivated by these examples, we study the allocation of indivisible items where relevance relationships between agents and items are represented via graphs.
In this framework, vertices correspond to agents, edges correspond to items, and an agent is incident to an edge if the item is relevant to that agent.
An agent can only receive the edges incident to her, and such an outcome is called an \emph{orientation}.

The allocation model with relevance relationships, also referred to as the orientation model in this paper, has been studied in machine scheduling \citep{DBLP:journals/algorithmica/EbenlendrKS14} and has recently gained popularity in fair division 
\citep{DBLP:conf/sigecom/0001FKS23,DBLP:journals/corr/abs-2409-13616,DBLP:journals/corr/abs-2407-05139,DBLP:conf/sigecom/AmanatidisFS24}. 
Most of these works on the fair division problem focus on envy-freeness, and to the best of our knowledge, the only exception is the recent work by \citet{DBLP:journals/corr/abs-2506-20317}, which investigates maximin share (MMS)  fairness \cite{budish2011combinatorial} within the orientation model.
However, proportionality, another well-studied fairness notion in the classical setting, has been largely overlooked.
Proportionality has been widely adopted as a golden fairness criterion in, for example, cake cutting \citep{DBLP:journals/geb/ChenLPP13,DBLP:journals/cacm/Procaccia13}, clustering \citep{DBLP:conf/icml/ChenFLM19,DBLP:conf/nips/Kellerhals024}, committee selection \citep{DBLP:conf/aaai/Kalayci0K24}, participatory budgeting \citep{DBLP:conf/nips/PierczynskiSP21}, matching \citep{DBLP:conf/nips/AmanatidisBFV22}, sortition \citep{DBLP:journals/nature/FlaniganG0HP21,DBLP:journals/corr/abs-2406-00913}, and bargaining \citep{kalai1977proportional}.

In the classic setting, an allocation is called \emph{proportional} (PROP) if, informally, each agent gets at least as much value as they would if all the resources were evenly distributed. 
That is, each agent receives a value of at least $\frac{1}{n}$ fraction of her total value of all items, where $n$ refers to the number of agents.
When the relevance relationship is introduced, we refine the threshold so that each agent expects her relevant items to be evenly distributed among the agents relevant to the underlying item.
Another fairness criterion that has not been studied in the orientation model is \emph{equitability} \citep{DBLP:books/daglib/0017730}, which requires that each agent receives the same level of value.

In this paper, we investigate the notions of proportionality and equitability, as well as their relaxations, under relevance relationships between agents and items. Our goal is to provide a picture of the existence of the computational results for these fairness notions.

\subsection{Our Contribution}

We consider the problem of allocating a set $E$ of $m$ indivisible items among $n$ agents, where items can be either goods (providing positive value) or chores (imposing negative value), or a combination of both. 
Each agent $i$ is associated with a set of relevant items $E_i \subseteq E$. 
We say an allocation is an orientation if each agent $i$ receives only the items relevant to her. Throughout the paper, we focus on the case where all items are allocated, and all resulting allocations are orientations.

\paragraph{Proportionality}

Suppose that agent $i$ has an additive valuation function $v_i(\cdot):2^{E}\rightarrow \mathbb{R}$, mapping a subset of items to a real number.
The original definition of the proportional share for agent $i$ is $\frac{1}{n}v_i(E)$, which is widely adopted in the literature.
To define proportionality in the orientation setting, one may set agent $i$'s value for each irrelevant item to 0, an approach applied in the literature (on envy-based notions, e.g., \citet{DBLP:journals/corr/abs-2409-13616}), 
and use the original definition of the proportional share.

However, in the orientation model, whether irrelevant items should be included in an agent's proportional share is a matter of discretion. For example, if a good $e$ is only relevant to agents $i$ and $j$, then $e$ would be allocated to one of them. If agent $i$ is aware that $e$ will never be allocated to the other $n-2$ agents, she might be unhappy if her entitlement is merely $\frac{1}{n}$ of $e$. Instead, she may reasonably expect to receive a $\frac{1}{2}$ share of $e$. Similarly, if a chore or task $e$ is relevant only to agents $i, j$, it would be overly optimistic for agent $i$ to assume she needs to undertake only a $\frac{1}{n}$ portion of $e$, as the task will ultimately be assigned to one of these two agents.

We now propose a refined definition of proportional share based on the relevance of each item.
For each item $e \in E$, let $n_e$ be the number of agents to whom $e$ is relevant. 
The proportional share of $i$ is 
\[
\PROP_i = \sum_{e \in E_i}\frac{1}{n_e}v_i(e),
\]
which means that, for each $e \in E_i$, agent $i$'s entitlement is $\frac{1}{n_e}$ of the value of $e$. This definition reduces to the original definition if each item $e$ is relevant to all agents.
An orientation is called PROP if every agent $i$'s value is at least $\PROP_i$.

Next, we summarize our results, beginning with PROP. First, the PROP orientations do not always exist. We are then interested in the problem of deciding whether an instance admits a PROP orientation and have the following results.
\begin{itemize}
    \item The decision problem is NP-complete when items are all goods (or all chores) and agents have {(1,2) or (-1,-2)} bi-valued valuations and the relevance can be represented by simple graphs, i.e., every item is relevant to exactly two agents and $|E_i\cap E_j| \le 1$ for all $i,j$.
    \item When valuations are binary (the absolute marginal value of any item being 1), the decision problem can be answered in polynomial time for arbitrary relevance. 
\end{itemize}

We then consider the up to one or any relaxation of PROP.
Informally, an orientation is PROP up to one (resp., any) item, abbreviated as PROP1 (resp., PROPX), if each agent would satisfy PROP after hypothetically adding or removing one (resp., any) item from her bundle.
For the weaker criterion of PROP1, we show that
\begin{itemize}
    \item PROP1 and fractional Pareto optimality (fPO) orientations exist and can be computed in polynomial time for mixed items and arbitrary relevance.
    \item
    For goods and multigraphs, there is always a PROP1 orientation (called strongly PROP1) that is also an MMS fair allocation, which in turn implies the existence of MMS allocations in the multigraph model, as proved in \citep{DBLP:journals/corr/abs-2506-20317}.

\end{itemize}

This result strengthens the unweighted result in \citep{DBLP:journals/orl/AzizMS20}, where each item can be allocated to every agent.
Moreover, we propose a stronger version of PROP1, called SPROP1, which cannot be guaranteed to exist for additive valuations without relevance constraints. However, such an orientation does exist in the multigraph orientation model.
We also consider PROPX, of which the results are
more negative.
\begin{itemize}
    \item PROPX orientations do not always exist.
    \item It is NP-complete to decide whether an instance admits a PROPX orientation, even for simple graphs containing only goods (or only chores) and agents with {(0,1)- or (0,-1)-}binary valuations. 
\end{itemize}

These results show a sharp contrast to the unconstrained setting, where PROPX allocations always exist for chores \citep{DBLP:conf/www/0037L022,DBLP:journals/ai/AzizLMWZ24}.

\paragraph{Equitability}

An allocation is equitable (EQ) if the subjective values of all agents are the same.
For indivisible items, EQ allocations do not always exist, and hence, we consider relaxations such as EQ1 and EQX. 
Their formal definitions are provided in the next section. In contrast to the results in the unconstrained setting, where EQ1 and EQX allocations are guaranteed to exist, we prove that
\begin{itemize}
    \item EQ1/EQX orientations do not always exist.
    \item It is NP-complete to decide whether an instance admits an EQ, EQ1, or EQX orientation, even when the relevance can be represented through simple graphs.
\end{itemize}

\paragraph{Envy-freeness}

Although envy-freeness has been studied within the orientation model in the literature, its relaxation, envy-free up to one item (EF1), remains poorly understood for chores. 
To address this gap, we complement the existing results by examining EF1 orientations for chores. 
Different from proportionality and equitability, envy-based fairness criteria are intrapersonal; that is, each agent has an evaluation on other agents' bundles so that whether envy exists becomes verifiable.
To ensure that envy-based notions are well-defined, a common assumption in existing work is that the value of any irrelevant item is zero for every agent \citep{DBLP:journals/corr/abs-2409-13616,DBLP:conf/ijcai/0027WL024}. In this paper, we adopt this assumption when examining envy-based notions of fairness.

Our results show that, in stark contrast to goods for which an EF1 orientation always exists even in the orientation model \citep{DBLP:journals/corr/abs-2409-13616}, EF1 orientations may not exist for chores, even when the relevance relationships are represented by simple graphs. 
Moreover, for such simple graphs, we characterize the necessary and sufficient conditions that guarantee the existence of EF1 orientations and show that it can be decided in polynomial time whether a given instance admits one. 
However, when relevance is modeled using multigraphs, i.e., each item is relevant to two agents, and multiple items can be relevant to a pair of agents, the decision problem becomes NP-complete.

\subsection{Other Related Work}

We refer to surveys \citep{moulin2019fair,DBLP:journals/ai/AmanatidisABFLMVW23} for comprehensive coverage of recent results on the fair allocation of indivisible items, and \citep{DBLP:journals/sigecom/Suksompong21} for various constraints that have been studied.
In the following, we recall the most relevant work to our paper. 
Without constraints, EF1  and PROP1 allocations always exist and can be computed in polynomial time for goods, chores, and the mixture of goods and chores \citep{DBLP:conf/approx/BhaskarSV21}.
Furthermore, PROP1 and Pareto optimal are known to be compatible in the mixed setting \citep{DBLP:conf/aaai/BarmanK19,DBLP:journals/orl/AzizMS20}. 
PROPX allocations exist for chores but not for goods \citep{DBLP:journals/orl/AzizMS20,DBLP:conf/www/0037L022,DBLP:journals/ai/AzizLMWZ24}.
The existence of EFX allocation remains unknown, except for several special cases \citep{DBLP:journals/jacm/ChaudhuryGM24}.

Initiated by \citet{DBLP:conf/sigecom/0001FKS23}, EFX orientations of relevant items have been extensively studied. 
\citet{DBLP:conf/ifaamas/AfshinmehrDKMR25} and \citet{DBLP:journals/corr/abs-2410-12039} extended the study to multi-graphs, and \citet{DBLP:conf/ijcai/0027WL024} and \citet{DBLP:journals/corr/abs-2501-13481} generalized the results to chores, the mixture of goods and chores, and other variants of the EFX concept.
\citet{DBLP:journals/corr/abs-2404-13527} characterized the graph structures for which EFX orientations always exist.
In contrast to EFX, when the items are goods, \citet{DBLP:journals/corr/abs-2409-13616} proved that an EF1 orientation always exists even in the general hypergraph orientation model. 
Without orientation constraints, 
\citet{DBLP:journals/corr/abs-2407-05139,DBLP:conf/sigecom/AmanatidisFS24,DBLP:journals/corr/abs-2506-09288} studied the computation of approximate EFX allocations in multi-graphs, where an edge may be allocated to non-incident edges. 
Recently, \citet{DBLP:journals/corr/abs-2506-20317} examined exact and approximate MMS allocations with additive and more general valuation functions.

\section{Preliminaries}
\subsection{The Orientation Model}
For any $k\in \mathbb{N}^+$, let $[k]=\{1,\ldots,k\}$. We study the model of allocating a set $E=\{e_1,\ldots,e_m\}$ of $m$ indivisible items to a set $N=\{1,\ldots,n\}$ of $n$ agents.
Each agent $i$ has an additive valuation function $v_i: 2^E \rightarrow \mathbb{R}$, that is, for any $S\subseteq E$, $v_i(S)=\sum_{e\in S}v_i(\{e\})$. 
For ease of notation, we write $v_i(e)$ instead of $v_i(\{e\})$. 
We call a valuation {\em binary} if $v_i(e) \in \{0,a\}$ for all items $e\in E$, where $a=1$ or $-1$.
For any $S\subseteq E$, let $|S|$ be the number of items in $S$.
An item is a good (resp., a chore) for an agent if it yields a non-negative (resp., non-positive) value. If the value of an item is zero for an agent, then the item can be a good or a chore.
In this paper, we say an instance is a \emph{goods-instance} (resp., \emph{chores-instance}) if all items are goods (resp., chores) for all agents.
We call it a \emph{mixed-instance} if an item can yield a positive value for one agent but a negative value for another.

In the general {\em orientation model}, each agent $i$ has a non-empty set of {\em relevant} items $E_i\subseteq E$. 
For each item $e\in E$, let $N_e =\{i \in N \mid e\in E_i\}$ be the set of agents to whom $e$ is relevant, and let $n_e = |N_e|$. 
We assume that $N_e$ is non-empty for every $e$, as otherwise, this item can be removed.
The orientation model is general and incorporates the classic unconstrained setting, where $E_i = E$ and $N_e=N$ for all $i\in N$ and $e\in E$.
Throughout the paper, we assume $v_i(e)=0$ for all $i$ and $e\in E\setminus E_i$.

We are also interested in two structured cases.
In the {\em simple graph} orientation model, $n_e = 2$ for all $e$ and $|E_i\cap E_j| \le 1$ for all $i\neq j$.
That is, the model can be described as a simple graph $G=(N,E)$, where each vertex is an agent and each edge is an item relevant to the two agents incident to this edge. 
Similarly, in the {\em multigraph} orientation model, $n_e = 2$ for every $e$, but the number of items in $E_i\cap E_j$ is not limited. 
That is, there may be multiple edges between any two agents (vertices).
In this paper, when the instance is a graph or multigraph, we use the terminologies vertex $i$ and agent $i$, and edge $e$ and item $e$, interchangeably.

\subsection{Solution Concepts}

In the orientation model, each item can only be allocated to an agent to whom the item is relevant. 
Formally, an {\em orientation} is denoted by $\pi = (\pi_1, \ldots, \pi_n)$, where $\pi_i \subseteq E_i$, and for any $i \neq j$, $\pi_i\cap \pi_j=\emptyset$ and $\bigcup_{i\in N} \pi_i=E$.
Below, we introduce proportional fairness for mixed instances. Let $\PROP_i=\sum_{e\in E_i} \frac{1}{n_e}\cdot v_i(e)$ be the (refined) proportional share for agent $i$.
\begin{definition}
    An orientation $\pi= (\pi_1, \ldots, \pi_n)$ is proportional (PROP) if for every agent $i\in N$, $v_i(\pi_i) \geq \PROP_i$.
\end{definition}

\begin{definition}
    An orientation $\pi= (\pi_1, \ldots, \pi_n)$ is proportional up to any item (PROPX) if for every agent $i\in N$, either (1) $v_i(\pi_i) \geq \PROP_i$, or (2)  $v_i(\pi_i\cup\{e\}) \geq \PROP_i$ for any $e \in E_i\setminus\pi_i$ such that $v_i(e)\geq 0$ and $v_i(\pi_i\setminus\{e\}) \geq \PROP_i$ for any $e\in\pi_i$ such that $v_i(e)\leq 0$.
\end{definition}

\begin{definition}
    An orientation $\pi= (\pi_1, \ldots, \pi_n)$ is proportional up to one item (PROP1) if for every agent $i\in N$, one of the following three holds: (1) $v_i(\pi_i) \geq \PROP_i$, (2)  $v_i(\pi_i\cup\{e\}) \geq \PROP_i$ for some item $e \in E_i\setminus\pi_i$, or (3) $v_i(\pi_i\setminus\{e\}) \geq \PROP_i$ for some item $e\in\pi_i$.
\end{definition}
A PROP orientation satisfies PROPX, which in turn satisfies PROP1. The above definition can be directly applied to goods- and chores- instances.
For example, for goods-instances, an orientation is PROPX if for every agent $i$, $v_i(\pi_i\cup\{e\}) \geq \PROP_i$ for any $e \in E_i\setminus\pi_i$;
For chores-instances, an orientation is PROPX if for every agent $i$, 
$v_i(\pi_i\setminus\{e\}) \geq \PROP_i$ for any $e\in\pi_i$. Regarding PROP1, the quantifier of $e$ in the prior two definitions is changed to existence. 

In this paper, we also consider equitability (EQ) and envy-freeness (EF). An orientation $\pi = (\pi_1, \ldots, \pi_n)$ is equitable (EQ) if $ {v_i}(\pi_i) =  {v_j}(\pi_j)$ for any two agents $i, j\in N$, and envy-free (EF) if $ {v_i}(\pi_i) \ge  {v_i}(\pi_j)$.
That is, in EQ orientations, the agents have the same value, while in EF orientations, every agent has the largest value in their own allocations. 
It is known that exact equitability or envy-freeness is not guaranteed to exist for indivisible items, and hence, we focus on their relaxations\footnote{EFX has been studied in, e.g.,
\cite{DBLP:conf/sigecom/0001FKS23,DBLP:conf/ijcai/0027WL024}, and thus we omit the corresponding discussion in the current paper.}.

\begin{definition}[EQX] 
An orientation 
$\pi = (\pi_1, \ldots, \pi_n)$ is equitable up to any item (EQX) if for any two agents $i, j\in N$, either (1) $v_i(\pi_i)=v_j(\pi_j)$, or (2) $ {v_i}(\pi_i) \geq  {v_j}(\pi_j\setminus\{e\})$ holds for every item $e\in \pi_j$ with $v_j(e)>0$, and $ {v_i}(\pi_i\setminus\{e\}) \geq  {v_j}(\pi_j)$ holds for every item $e\in \pi_i$ with $v_i(e) <0$.
\end{definition}

\begin{definition}[EQ1] 
An orientation $\pi = (\pi_1, \ldots, \pi_n)$ is equitable up to one item (EQ1) if for any two agents $i, j\in N$, one of the following three holds: (1) $v_i(\pi_i)=v_j(\pi_j)$, (2) there exists $e\in \pi_j$ such that $ {v_i}(\pi_i) \geq  {v_j}(\pi_j\setminus\{e\})$, or (3) there exists $e\in \pi_i$ such that $ {v_i}(\pi_i\setminus\{e\}) \geq  {v_j}(\pi_j)$.
\end{definition}

\begin{definition}[EF1] 
An orientation $\pi = (\pi_1, \ldots, \pi_n)$ is envy-free up to one item (EF1) if for any two agents $i, j\in N$, one of the following three holds: (1) $v_i(\pi_i)\ge v_i(\pi_j)$, (2) there exists $e\in \pi_j$ such that $ {v_i}(\pi_i) \geq  {v_i}(\pi_j\setminus\{e\})$, or (3) there exists $e\in \pi_i$ such that $ {v_i}(\pi_i\setminus\{e\}) \geq  {v_i}(\pi_j)$.
\end{definition}

In the relevance model, an item that is not relevant to an agent is assigned a value of 0 for that agent \citep{DBLP:journals/corr/abs-2409-13616,DBLP:conf/ijcai/0027WL024}.
While this approach is reasonable for goods, it becomes questionable for chores, as we may wish to allocate items that impose no cost on the agents.
In this sense, we argue that proportionality and equitability are more suitable fairness criteria in this context, since their definitions do not depend on an agent’s valuation of irrelevant items.

\section{Proportionality and Its Relaxations}
\label{sec::prop:meta}

In this section, we study PROP fairness and its relaxations, namely PROP1 and PROPX, and present results on their existence and on addressing the computational complexity of the associated decision problems. 

\subsection{PROP Fairness}\label{sec::prop}

Similar to the classic fair division of indivisible items, PROP is not always satisfiable even for simple graph instances. Consider an item and two agents, and the item is relevant and yields positive value to both agents.
A natural question is whether one can efficiently compute a PROP orientation when it exists.
Unfortunately, the answer is no even for restricted instances. 
In the following, if agents' total value of their relevant items is identical, we say the valuations are \emph{normalized}.

We derive the reduction from 2P2N-3SAT problem, known to be NP-complete \cite{DBLP:journals/dam/BermanKS07,DBLP:conf/rta/Yoshinaka05}. A 2P2N-3SAT instance contains a Boolean formula in conjunctive normal form consisting of the set of variables $X=\{x_1,\ldots,x_s\}$ and the set of clauses $C=\{C_1,\ldots,C_t\}$.
Each variable appears exactly twice as a positive literal and exactly twice as a negative literal in the formula, and each clause contains three distinct literals, i.e., $3t=4s$. Denote by $L=\bigcup_{j=1}^s\{x_j^1,x_j^2,\bar{x}_j^1,\bar{x}_j^2\}$ the set of literals and by $C(\ell)$ the clause that contains the literal $\ell$.

\begin{theorem}\label{thm::prop-complexity}
     Deciding the existence of PROP orientations is NP-complete, even for simple graphs where (1) all items are goods, valuations are normalized, and $v_i(e)\in \{1,2\}$ for all $i$ and $e\in E_i$; or all items are chores, valuations are normalized, and $v_i(e)\in \{-1,-2\}$ for all $i$ and $e\in E_i$.
\end{theorem}
\begin{proof}[Proof for the goods-instance]
The problem is in NP, as given an orientation, one can verify whether it is PROP or not in polynomial time.
Next, we derive the NP-completeness by a reduction from 2P2N-3SAT.

Given a 2P2N-3SAT instance, we create a goods-instance as follows:
\begin{itemize}
    \item for each clause $C_j$, create a vertex $c_j$ and 11 dummy vertices $d^1_j,\ldots, d^{11}_j$. Create edge $(d_j^1,c_j)$ with value 1 for $d_j^1$ and $c_j$. Moreover, the construction of edges with two endpoints being dummy vertices and their values are illustrated in Figure~\ref{fig:clause_gadget};
    \item for each clause $C_j$, create 3 edges based on the following rule: if $C_j$ contains literal $x_i^k$ (resp., $\bar{x}_i^{k'}$), create edge $(c_j, i)$ (resp., $(c_j, \bar{i})$) with value 1 for both incident agents, as illustrated in Figures \ref{fig:clause_gadget} and \ref{fig:variable_gadget} (edges without an endpoint). 
	  \item for each variable $x_i$, create two vertices $i$ and $\bar{i}$ and one edge $(i,\bar{i})$ with value 2 to both of them, as illustrated in Figure~\ref{fig:variable_gadget};
\end{itemize}
The created instance has $2s+12t$ vertices and $s+17t$ edges. For any $i$ and $e\in E_i$, $v_i(e)\in \{1,2\}$ holds. Moreover, each agent has a total value of 4 for her incident edges (normalized valuations), and thus, the proportional share of each agent is 2.

\begin{figure}[htbp]
    \centering
    \begin{subfigure}[t]{0.45\textwidth}
        \centering
        \includegraphics[width=\textwidth]{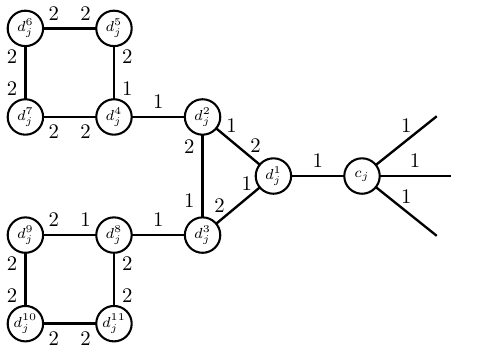}
        \caption{The illustration for goods-instance of the clause gadget for $C_j$. For each edge, if there is only one label, it represents the value of the edge to both endpoints. If there are two labels, the one closer to a vertex represents that vertex's value for the edge.}
        \label{fig:clause_gadget}
    \end{subfigure}
    \hfill
    \begin{subfigure}[t]{0.4\textwidth}
        \centering
        \includegraphics[width=\textwidth]{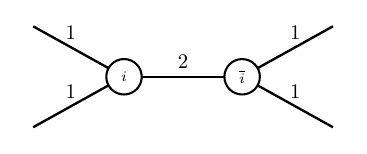}
        \caption{The illustration for goods-instance of the variable gadget for $x_i$. The label on each edge represents the value of the edge to both endpoints}
        \label{fig:variable_gadget}
    \end{subfigure}
    
    \caption{Combined illustration of the clause gadget and variable gadget.}
    \label{fig:combined_gadgets}
\end{figure}

Suppose that there exists a truth assignment that satisfies all clauses in $C$. For $c_j,d^1_j,\ldots,d^{11}_j$, observe that we can satisfy each $d^r_j$ with her proportional share even when $(d^1_j,c_j)$ is allocated to $c_j$; 
allocate edges in the cycle with 3 vertices $d^1_j,d^2_j,d^3_j$ in a clockwise manner, and allocate edges in the cycles with 4 vertices in an anticlockwise manner. 
Edges $(d^2_j,d^4_j)$ and $(d^3_j,d^8_j)$ are arbitrarily allocated to their incident vertices.
Thus, we create a partial assignment where each $d^r_j$ receives a value of at least 2 and each $c_j$ receives 1.

For each variable $x_i$, if $x_i$ is \texttt{True}, allocate $(i,\bar{i})$ to vertex $i$. Then allocate the other two edges incident to $i$ to the vertices corresponding to $C(x^1_i)$ and $C(x^2_i)$; recall that $C(x^1_i)$ refers to the clause that contains literal $x^1_i$. 
For vertex $\bar{i}$, allocate her the two edges with value 1. Similarly, if $\bar{x}_i$ is \texttt{True}, allocate $(i,\bar{i})$ to vertex $\bar{i}$ and allocate the other two edges incident to $\bar{i}$ to the vertices corresponding to $C(\bar{x}^1_i)$ and $C(\bar{x}^2_i)$, and for vertex $i$, allocate her the two edges with value 1.
At this point, for any $i\in [s]$, both $i$ and $\bar{i}$ receive value 2.

For each vertex $c_j$, as clause $C_j$ is satisfied, $c_j$ receives one additional incident edge besides $(d^1_j,c_j)$, resulting in a value of at least 2. Therefore, we create a PROP orientation.

Next, for the reverse direction, suppose that there exists a PROP orientation $\pi$. We now create a truth assignment of $\{x_i\}_{i\in [t]}$ as follows: if $(i,\bar{i})$ is allocated to vertex $i$, then set $x_i$ to \texttt{True}; and otherwise, if $(i,\bar{i})$ is allocated to vertex $\bar{i}$, set $\bar{x}_i$ to \texttt{True}.
Such a truth assignment ensures that exactly one of $x_i,\bar{x}_i$ is set to \texttt{True}, and hence, the truth assignment is valid. Next, we prove that the assignment satisfies all clauses.

For a contradiction, suppose that there exists a clause $C_{j'}$ not satisfied. By the created truth assignment, each of the three vertices corresponding to the three literals in $C_{j'}$ does not receive the edge with value 2.
Consequently, the three edges connecting $c_{j'}$ to variable vertices must not be allocated to $c_{j'}$ so that the three vertices corresponding to the three literals in $C_{j'}$ can satisfy PROP.
Hence, the value of $c_{j'}$ is at most 1, meaning that the orientation is not PROP, deriving the desired contradiction. Therefore, the created truth assignment satisfies all clauses.

As for the chores-instance, we use the same graph construction, but each agent's values on the edges are mapped to the corresponding negative values. 
If an edge has value 1 (resp., 2) for a vertex in the reduction for goods, we now change it to $-1$ (resp., $-2$). For any $i$ and $e\in E_i$, $v_i(e)\in \{-1,-2\}$ holds. Moreover, each agent has a total value of $-4$ for her incident edges (normalized valuations), and thus, the proportional share of each agent is $-2$.

Suppose there is a truth assignment that satisfies all clauses in $C$. For $c_j,d^1_j,\ldots,d^{11}_j$, observe that we can satisfy each $d^r_j$ her proportional share even when $(d^1_j,c_j)$ is allocated to $d^1_j$; 
allocate edges in the cycle with 3 vertices $d^1_j,d^2_j,d^3_j$ in an anticlockwise direction, and allocate edges in cycles with 4 vertices in a clockwise direction.
Edges $(d^2_j,d^4_j)$ and $(d^3_j,d^8_j)$ are arbitrarily allocated to their incident vertices.
At this point, each $d^r_j$ satisfies her proportional share, and her value is unchanged thereafter.

For each $x_i$, if $x_i$ is \texttt{True}, allocate $(i,\bar{i})$ to $\bar{i}$, and allocate $i$ the other two edges incident to her; if $\bar{x}_i$ is \texttt{True}, allocate $(i,\bar{i})$ to $i$, and allocate $\bar{i}$ the other two edges incident to her. 
At this point, for any $i\in [s]$, both $i$ and $\bar{i}$ receive value $-2$.
For each vertex $c_j$, since clause $C_j$ is satisfied, at least one incident edge other than $(d^1_j, c_j)$ is not allocated to her, which implies that her value is at least $-2$. Therefore, the created orientation is PROP.

Next, for the reverse direction, suppose that there exists a PROP orientation $\pi$. We create a truth assignment of $\{x_i\}_{i\in[t]}$ as follows: if $(i,\bar{i})$ is allocated to vertex $i$, then set $\bar{x}_i$ to \texttt{True}, and otherwise, set $x_i$ to \texttt{True}.
Such a truth assignment ensures that exactly one of $x_i,\bar{x}_i$ is set to \texttt{True}. Hence, the truth assignment is valid. 
Moreover, if some clause $C_{j'}$ is not satisfied, then in $\pi$, $c_{j'}$ receives a value at most $-3$, violating PROP. Therefore, the created truth assignment satisfies all clauses.
\end{proof}

To complement the hardness result, we show that for binary valuations and goods-instance (or chores-instance), one can compute in polynomial time a PROP orientation when it exists.

\begin{theorem}\label{thm::prop-compute-binary}
    For the general orientation model with binary additive valuations, one can in polynomial time determine whether a PROP orientation exists or not, and compute one if it exists.
\end{theorem}
\begin{proof}
	We begin with the goods-instance. If some item has a value of zero for all relevant agents, arbitrarily allocate the item to the relevant agents. 
	If some item has a value of one for exactly one relevant agent, allocate the item to that agent.
	At this point, let $\pi'$ be the current partial assignment, and for each $i$, let $a_i$ be the current value of agent $i$. Moreover, each of the unallocated items has a value of one for at least two agents relevant to that item. 
	
	Next, create a bipartite graph with two parts of vertices $X$ and $Y$. For each unallocated item with respect to $\pi'$, create a vertex in $X$. For each agent $i$, create a number $\lceil \PROP_i\rceil - a_i$ of vertices in $Y$. 
	For each such a vertex, connect it to the vertex $x\in X$ if agent $i$ is relevant to the item corresponding to $x$ and has value one for it. We prove below that a PROP orientation exists if and only if a $Y$-perfect matching exists in the created bipartite graph.
	
	For the ``if'' direction, based on the $Y$-perfect matching, one can extend $\pi'$ to another (possibly partial) assignment where each agent $i$ receives her proportional share. Moreover, unallocated items (if any) make agents weakly better off as items are goods. Thus, a PROP orientation exists.
	For the ``only if'' direction, denote by $Q\subseteq E$ the set of items with value one for at least two of their relevant agents. In the PROP orientation, each agent $i$ must receive value at least $\lceil \PROP_i \rceil - a_i$ when restricting to $Q$ as items $E\setminus Q$ can increase the value of agent $i$ by at most $a_i$.
	Hence, one can convert the assignment of $Q$ in the PROP to a $Y$-perfect matching.
	
	By applying the matching algorithm, one can determine whether a PROP orientation exists or not and compute one when it exists in time polynomial in $n$ and $m$ .
	
	As for the chores-instance, the idea is similar. First, if some item has a value of zero for some relevant agent, then allocate it to that agent. At this point, each of the unallocated items has a value of $-1$ for every relevant agent.
	Next, we create a bipartite graph $(X\cup Y)$. For each unallocated item, create a vertex in $X$.
	For each agent $i$, create a number $\lfloor |\PROP_i| \rfloor$ of vertices in $Y$. For each such agent vertex, connect it to the vertex $x\in X$ if agent $i$ is relevant to the item corresponding to $x$.
	By the arguments similar to those for the goods-instance, one can verify that a PROP orientation exists if and only if there exists a $X$-perfect matching in the created bipartite graph.
	Moreover, when an $X$-perfect matching exists, we can convert the matching to a PROP orientation.
\end{proof}

\subsection{PROP1 and PROPX Fairness}\label{sec::prop1}

Motivated by the above impossibility results, we now consider the relaxations of PROP.
We begin with PROP1 and show that PROP1 orientations always exist and are compatible with fractional Pareto optimal (fPO). 
An orientation $\pi=(\pi_1,\ldots,\pi_n)$ \emph{Pareto dominates} another orientation $\pi'=(\pi'_1,\ldots,\pi'_n)$ if for any $i\in [n]$, $v_i(\pi_i) \geq v_i(\pi'_i)$ and at least one inequality is strict.
An orientation $\pi$ is Pareto optimal (PO) if no integral orientation Pareto dominates it, and is fPO if no fractional orientation Pareto dominates it, where in an integral orientation, each item is fully allocated to one agent, and in a fractional orientation, a fraction of item can be assigned to some agent.
Formally, in a fractional orientation $\pi=(\pi_1,\ldots,\pi_n)$, $\pi_i=(\pi_{i,e})_{e\in E}$ where $0 \le \pi_{i,e} \le 1$ represents the portion of $e$ allocated to agent $i$.

\begin{theorem}
\label{thm:prop1:fPO}
For the general orientation model and the mixed-instance, we can compute a PROP1 and fPO orientation in polynomial time.
\end{theorem}
    We start from a fractional proportional orientation created as follows: for any $e$ and any $i \in N_e$, let $\pi_{i,e} = \frac{1}{n_e}$. In this fractional orientation, each agent receives her proportional share.
Then we compute another fPO fractional orientation that Pareto improves the initial proportional orientation. Moreover, the new fractional orientation has an acyclic \emph{undirected consumption} graph: a bipartite graph in which vertices on one side are the agents, vertices on the other side are the items, and there is an edge between agent $i $ and item $e$ if and only if $\pi_{i,e}>0$.
To show the existence of such an fPO fractional orientation, we adapt the techniques in \cite{DBLP:journals/ior/SandomirskiyS22} to the orientation model.
Last, we round the fractional orientation to a PROP1 and fPO integral orientation.
The formal proof can be found in Appendix \ref{sec::prop:meta}.

We next study PROPX, a notion stricter than PROP1.
In sharp contrast, the results for PROPX are mostly negative. Particularly, a PROPX orientation does not always exist, even in the case of simple graphs with binary valuations, and determining whether such an orientation exists is computationally intractable.

\begin{proposition}
   PROPX orientations may not exist, even for simple graphs and binary valuations.
\end{proposition}
\begin{proof}
    Let us begin with the goods-instance and consider the instance illustrated in Figure~\ref{fig:PROPX+no}. Fix an arbitrary orientation $\pi$, and due to symmetry, assume that $(1,2)$ is allocated to agent 2. Now we focus on agents 1,3, and 4. 
    To ensure that agent 1 satisfies PROPX, both edges $(1,3)$ and $(1,4)$ must be allocated to agent 1. As a consequence, one of agent 3 and agent 4 receives $\emptyset$ and violates PROPX.
    
    \begin{figure}[H]
    \centering
    \includegraphics[width=0.32\linewidth]{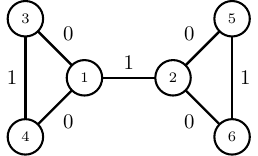}
    \caption{The illustration of the goods-instance where PROPX orientations do not exist. The label on each edge represents the value of the edge to both endpoints}
    \label{fig:PROPX+no}
\end{figure}
    As for chores-instance, we consider the same graph construction and transform agents' valuations into their negatives. By arguments similar to those of goods, one can verify that no orientation is PROPX.
\end{proof}

\begin{theorem}\label{thm::complexity-propx}
    For both goods- and chores-instances, deciding the existence of PROPX orientations is NP-complete, even for simple graphs and binary valuations.
\end{theorem}
\begin{proof}
The decision problem is clearly in NP. We next derive the NP-completeness by a reduction from 2P2N-3SAT. 

Given a 2P2N-3SAT instance, we create a goods-instance as follows:
\begin{itemize}
    \item for each clause $C_j$, create a vertex $c_j$ and 6 dummy vertices $d_j^1,\ldots, d_j^6$. Create edges $(c_j,d^1_j)$ and $(c_j,d^4_j)$; moreover, create edges with both endpoints being dummy vertices, as illustrated in Figure \ref{fig:PROPX+clause};
    \item for each clause $C_j$, create three more edges: if $C_j$ includes literal $x_i^k$ (resp., $\bar{x}_i^{k'}$), add an edge $(c_j,i)$ (resp. $(c_j,\bar{i})$ with value 1 for both incident agents, as illustrated in Figure \ref{fig:PROPX+clause}.
    \item for each variable $x_i$, create two vertices $i$ and $\bar{i}$, and an edge $(i,\bar{i})$ with value $1$ for both $i,\bar{i}$. 
    Moreover, apart from $(i,\bar{i})$, each $i$ and $\bar{i}$ values all other incident edges at 0.
\end{itemize}

\begin{figure}
\centering
        \centering
        \includegraphics[width=0.4\linewidth]{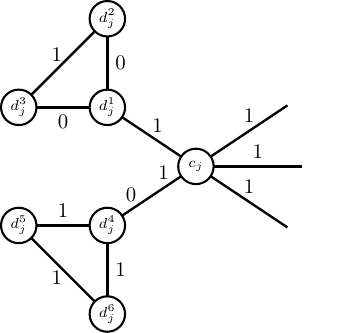} 
        \caption{The illustration for goods-instance of the clause gadget for $C_j$.
        For each edge, if there is only one label, it represents the value of the edge to both endpoints. If there are two labels, the label closer to a vertex represents that vertex's value for the edge.}
        \label{fig:PROPX+clause}
\end{figure}

The created instance has $2s+7t$ vertices and $s+11t$ edges; moreover, each agent's valuation is binary. 
We first claim that in a PROPX orientation (if it exists), for each $c_j$, edge $(c_j,d^1_j)$ must be allocated to $d^1_j$.
For a contradiction, assume that this is not the case. We now focus on agents $d^1_j,d^2_j,d^3_j$. Without loss of generality, assume edge $(d^2_j,d^3_j)$ is allocated to $d^2_j$ in the orientation.
Then in order to make $d^3_j$ satisfy PROPX, we need to allocate $(d^1_j,d^3_j)$ to her, which then makes $d^1_j$ violate PROPX, as $d^1_j$ has value zero in the orientation and edge $(d^1_j,d^3_j)$ (with value zero) is not allocated to her but her proportional share is $\frac{1}{2}$.

We now show that there exists a truth assignment satisfying all clauses if and only if there exists a PROPX orientation. 
Suppose that there exists a truth assignment satisfying all clauses. If $x_i$ is set to be \texttt{True}, we allocate $(i,\bar{i})$ to $i$. Then, we allocate to $\bar{i}$ all the edges that are currently unallocated and incident to her. Moreover, assign the two unallocated edges incident to $i$ to the vertices corresponding to $C(x_i^1)$ and $C(x_i^2)$.
Similarly, if $\bar{x}_i$ is set to be {\texttt{True}}, we orient these edges in the reverse direction. Then agents $i$ and $\bar{i}$ satisfy PROPX.

For each vertex $c_j$ and its dummy vertices, allocate $(c_j,d^1_j)$ to $d^1_j$ and allocate $(c_j,d^4_j)$ to $c_j$. One can allocate the edge with both endpoints being dummy vertices in a way such that each $d_j^r$ satisfies PROPX.
At this point, $c_j$ should receive at least one more edge to achieve PROPX. Note that the truth assignment satisfying clause $C_j$, and thus, the edge connecting $c_j$ to the vertex corresponding the true literal in $C_j$ is allocated to $c_j$.
Therefore, the created orientation is PROPX.

Next, for the reverse direction, suppose that there exists a PROPX orientation $\pi$. We now create a truth assignment as follows: if $(i,\bar{i})$ is allocated to vertex $i$, we set $x_i$ to \texttt{True}; otherwise we set $\bar{x}_i$ to \texttt{True}. 
Such a truth assignment ensures that exactly one of
$x_i$ and $\bar{x}_i$ is set to \texttt{True}, and hence, the truth assignment is valid. 
For a contradiction, we suppose that there exists a clause $C_j$ that is not satisfied. Then since we have proved that $(c_j,d^1_j)$ must be allocated to $d^1_j$, vertex $c_j$ has value at most 1 in $\pi$, as $C_j$ is not satisfied. Accordingly, $c_j$ does not satisfy PROPX, deriving the desired contradiction.

Next, we prove for the case of chores. Create the graph identical to that for the goods-instance. For agents' valuations, let edge $(c_j,d^1_j)$ result in value $-1$ for $d^1_j$ and value zero for $c_j$ for all $j$.
    For each of the other edges, if their value for a vertex is one in the goods-instance, change it to $-1$; if their value for a vertex is zero in the goods-instance, it remains the same.
    
    We first claim that in a PROPX orientation (if it exists), for each $c_j$, edge $(c_j, d^1_j)$ must be allocated to $c_j$. Assume for the contradiction that this is not the case. 
    We now focus on agents $d^1_j,d^2_j,d^3_j$. Without loss of generality, assume edge $(d^2_j,d^3_j)$ is allocated to $d^2_j$ in the orientation.
    Then in order to make $d^2_j$ satisfy PROPX, we need to allocated $(d^2_j,d^1_j)$ to $d^1_j$, which means that $d^1_j$ cannot receive $(c_j, d^1_j)$.
    Then in a PROPX orientation, the value of $c_j$ must be at most $-2$ as her proportional share is $-2$ and she always receives an edge with value zero. We prove below that there exists a truth assignment satisfying all clauses if and only if there exists a PROPX orientation.

    Suppose that there exists a truth assignment satisfying all clauses. If $x_i$ is set to \texttt{True}, we allocate $(i,\bar{i})$ to $\bar{i}$. Then we allocate $i$ the two edges with value zero for her and allocate the two unallocated edges incident to $\bar{i}$ to the vertices corresponding to $C(\bar{x}^1_i)$ and $C(\bar{x}^2_i)$.
    Similarly, if $\bar{x}_i$ is set \texttt{True}, we orient these edges in the reverse direction. Then at this point, both agents $i$ and $\bar{i}$ satisfy PROPX.

    For each vertex $c_j$ and its dummy vertices, allocate $(c_j,d^1_j)$ to $c_j$ and allocate $(c_j,d^4_j)$ to $d^4_j$. One can allocate the edges with both endpoints being dummy vertices in a way such that each $d^r_j$ satisfies PROPX. 
    Then all edges are allocated and $c_j$ has value at least $-2$ as the edge connecting $c_j$ to the vertex corresponding to the true literal in $C_j$ is not allocated to $c_j$. Therefore, the created orientation is PROPX.

    Next for the reverse direction, suppose that there exists a PROPX orientation $\pi$. We now create a truth assignment as follows: if $(i,\bar{i})$ is allocated to $i$, set $\bar{x}_i$ to \texttt{True}; otherwise, set $x_i$ to \texttt{True}.
    Such a truth assignment ensures that exactly one of $x_i$ and $\bar{x}_i$ is set to \texttt{True}, and hence, the truth assignment is valid.
    For a contradiction, suppose that there exists a clause $C_j$ that is not satisfied. Then $c_j$ has value at most $-3$ as each edge connecting $c_j$ to the variable vertices with literals in $C_j$ must be allocated to $c_j$.
    As edge $(c_j,d^1_j)$ is allocated to $c_j$, vertex $c_j$ violates PROPX, a contradiction.
\end{proof}

\subsection{Strongly PROP1 and MMS Fairness}

In this section, we focus on multigraph orientation and goods-instances.
We first introduce a stronger version of PROP1. 

\begin{definition}
    For goods-instances, an orientation $\pi= (\pi_1, \ldots, \pi_n)$ is strongly proportional up to one item (SPROP1) if for every agent $i\in N$, there exists $e^*\in E_i$ such that
    \[
    v_i(\pi_i) \geq \frac{1}{2}\sum_{e\in E\setminus\{e^*\}} v_i(e).
    \]
\end{definition}
For additive valuations without relevance constraints, it is known that an SPROP1 allocation may not always exist.
However, as we will prove below, an SPROP1 orientation does exist in the multigraph orientation model. 
Consider the following algorithm: 
\begin{itemize}
    \item Start with an arbitrary vertex $i$ with $E_i\neq \emptyset$, and allocate edge $e_{i,j} \in\arg\max_{e\in E_i}v_i(e)$ to vertex $i$. Remove $e_{i,j}$ from $E_i$ and $E_j$.
    
    
    \item Repeat the above procedure for vertex $j$, if $E_{j}\neq \emptyset$, until all edges are allocated.
    If $E_{j} = \emptyset$, the vertex can be chosen arbitrarily. 
\end{itemize}
That is, for every agent $i$, if one of her incident edges is allocated to some other agent, she can immediately receive the best remaining edge.
Thus, 
\[
    v_i(\pi_i) \geq \frac{1}{2}\sum_{e\in E\setminus\{e^*\}} v_i(e),
\]
where $e^* \in \arg \max_{e\in E_i} v_i(e)$ and the orientation is SPROP1.
It is proved in \cite{DBLP:journals/corr/abs-2506-20317} that this algorithm returns an MMS fair allocation. 
We next show that SPROP1 orientation is a stronger fairness notion than MMS allocations \cite{DBLP:journals/corr/abs-2506-20317}. 
Let $\Pi_n(E)$ be the set of all $n$-partitions of $E$.
Note that a partition in $\Pi$ may not correspond to an orientation.  
The maximin share (MMS) of agent $i\in N$ is defined as
\[
\MMS_i=\max_{(A_1,\ldots,A_n) \in \Pi_n(E)}\min_{1\le j\le n} v_i(A_j).
\]
An allocation $(X_1,\ldots,X_n)$ is MMS fair if $v_i(X_i) \ge \MMS_i$ for every agent $i \in N$.

\begin{proposition}
    For the general orientation model and the goods-instance with $n\ge 3$, every SPROP1 orientation is an MMS allocation.
\end{proposition}
\begin{proof}
    To prove the proposition, it suffices to prove 
    \[
    \frac{1}{2}\sum_{e\in E\setminus\{e^*\}} v_i(e) \ge \MMS_i,
    \]
    where $e^* \in \arg \max_{e\in E_i} v_i(e)$.
    If $v_i(e^*) \le \MMS_i$, then 
    \[
    \frac{1}{2}\sum_{e\in E\setminus\{e^*\}} v_i(e) \ge \frac{1}{3}\sum_{e\in E} v_i(e)\ge \frac{1}{n}\sum_{e\in E} v_i(e) \ge \MMS_i,
    \]
    If $v_i(e^*) > \MMS_i$, then 
    \[
    \MMS_i \le \frac{1}{n-1}v_i(E\setminus \{e^*\})\le  \frac{1}{2}v_i(E\setminus \{e^*\}).
    \]
    Recall that $n\ge 3$ in the above two inequalities, and thus the proposition is proved.
\end{proof}

\section{Equitability and Its Relaxations}
\label{sec:EQ}

In this section, we focus on equitability-related notions that have been widely studied in the unconstrained setting \citep{DBLP:conf/ijcai/FreemanSVX19,DBLP:journals/aamas/SunCD23,DBLP:conf/aaai/BarmanBPP24,DBLP:journals/eor/SunC25}, but remain underdeveloped in the orientation model.

In the unconstrained setting, EQX/1 allocations can be achieved through the greedy algorithm. However, our results indicate that neither EQ1 nor EQX is guaranteed to be satisfiable in the orientation model. Even worse, deciding their existence is NP-complete.
For completeness, we also provide the complexity results on the notion of EQ.

\begin{theorem}\label{thm::EQ-ak}
    Determining whether an EQ orientation exists on a simple graph in both goods- and chores-instance is NP-complete.
\end{theorem}
\begin{proof}
    We derive reductions from the \ssc{Partition} problem: given a set $\{x_1,\ldots,x_n\}$ of $n$ integers whose sum is $2T$, can $[n]$ be partitioned into two sets $I_1$ and $I_2$ such that $\sum_{j\in I_1} x_j=\sum_{j\in I_2}x_j = T$?

    The decision problem is in NP as given an orientation, one can decide whether it is EQ in polynomial time. We derive the reduction from the \ssc{partition} problem and begin with the goods-instance.
\begin{figure}
  \centering
\includegraphics[width=0.35\textwidth]{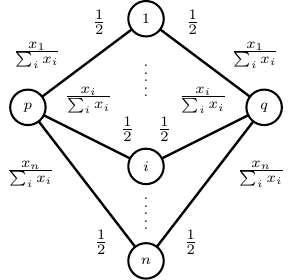}
  \caption{The illustration of the goods-instance for the reduction from EQ. The label closer to a vertex represents that vertex’s value for the edge.}\label{fig:EQ}
\end{figure}
We create a simple graph with $n+2$ vertices $\{1,\ldots,n,p,q\}$ and $2n$ edges. The created graph and the valuations of agents are illustrated in Figure \ref{fig:EQ}. One can verify that the instance is normalized as the total value of items for each agent is equal to one.

We claim that in an EQ orientation, each vertex $i$ with $i\in [n]$ must receive exactly one incident edge. If some vertex $i'\in [n]$ receives two incident edges, the value of agent $i'$ is one, but the value of $p$ is less than 1, violating EQ.
If some vertex $i'\in [n]$ receives no edge, then the value of agent $i'$ is zero, but the value of $p$ is positive, violating EQ.
Thus in an EQ orientation, the value of each agent must be $\frac{1}{2}$. Then the set of items received by $p$ and $q$ in an EQ orientation map a solution of the \ssc{partition} problem. Therefore, an EQ orientation exists if and only if the answer to the \ssc{partition} instance is a yes-instance.

For the reduction of the chores-instance, we create the same graph but negate agents' valuations. By similar arguments, we can prove that in an EQ orientation, each vertex $i\in [n]$ must receive exactly one incident edge. Therefore, an EQ orientation exists if and only if the answer to the \ssc{partition} instance is a yes-instance.
\end{proof}

Next, we show that EQ1 orientations do not always exist, which presents a sharp contrast to the classical setting without relevance constraints where the existence of EQ1 allocations is guaranteed.

\begin{proposition}
    An EQ1 orientation is not guaranteed to exist even for simple graphs.
\end{proposition}
\begin{proof}
    Consider the example with identical valuations shown in Figure \ref{fig:EQ1-nonexist}. 
    \begin{figure}[h]
  \centering
  \includegraphics[width=0.4\linewidth]{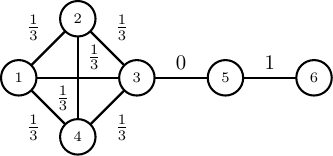}
  \caption{The illustration of the goods-instance where EQ1 does not exist. The label on each edge represents the value of the edge to both endpoints.}
  \label{fig:EQ1-nonexist}
\end{figure}
    In any orientation, one of the agents $5$ and $6$ receives a value of 0; suppose that this agent is $5$. However, as there are 6 edges among $1$, $2$, $3$, and $4$, by the pigeonhole principle, at least one of them will receive two or more edges. Suppose this agent is $1$, then we have $v_1(\pi_1\setminus\{e\}) = \frac{1}{3} > v_5(\pi_5) = 0$ for all $e\in \pi_1$, which means that EQ1 does not hold between $1$ and $5$. 
    As for the chores-instance, we use the same graph while changing each agent's value on each edge to the corresponding negative value. In that case, we have $v_1(\pi_1\setminus \{e\}) = -\frac{1}{3} < v_5(\pi_5) = 0$ for all $e\in \pi_1$, which violates EQ1.
\end{proof}

We also investigate the problem of deciding the existence of EQ1 orientations for a given instance. This decision problem proves to be computationally intractable.

\begin{theorem}\label{thm::EQ1-ak}
Determining whether an EQ1 orientation exists on a simple graph in both goods-
and chores-instance is NP-complete.
\end{theorem}
\begin{proof}
    \begin{figure}[h]
    \centering
\includegraphics[width=0.65\textwidth]{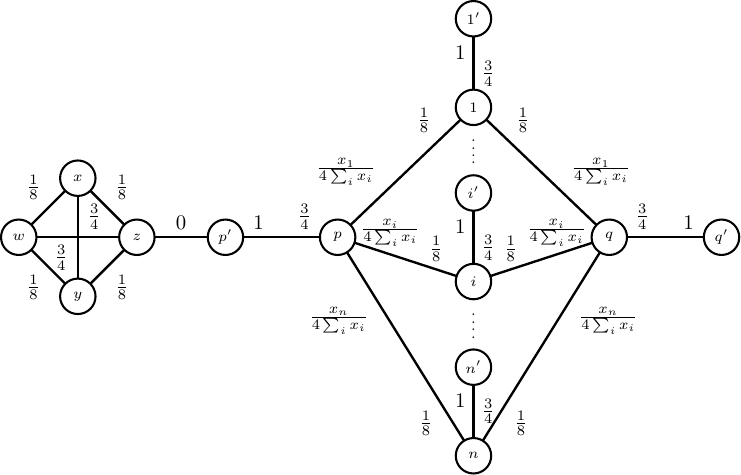}
\caption{Illustration of the goods-instance for reduction from EQ1. For the edge with one label, the label is the value of that edge for both endpoints. If an edge has two labels, the label closer to a vertex represents that vertex’s value for the edge.}
    \label{fig:EQ1NPC}
\end{figure}

The decision problem is in NP as given an orientation, one can decide whether it is EQ1 in polynomial time. We derive the reduction from the \ssc{partition} problem and begin with the goods-instance.

We create a goods-instance of a simple graph, as illustrated in Figure \ref{fig:EQ1NPC}:
\begin{itemize}
    \item create vertices $1,\ldots,n$ and vertices $1',\ldots,n'$, and for each $i\in [n]$, create an edge $(i,i')$;
    \item create vertices $p,p',q,q'$, and create edges $(p,p')$ and $(q,q')$. Moreover, for each $i\in [n]$, create edges $(p,i)$ and $(q,i)$;
    \item create vertices $w,x,y,z$ and edges such that these four vertices form $K_4$ (i.e., a complete graph with 4 vertices). Moreover, create edge $(z,p')$.
\end{itemize}
We define the agents' valuations as follows:

\begin{itemize}
        \item for each agent $i\in [n]$, ${v_i}((i,i')) = \frac{3}{4}$, $ {v_{i'}}((i,i')) = 1$, $ {v_i}((p,i)) =  {v_i}((q,i)) = \frac{1}{8}$, and ${v_p}((p,i)) =  {v_q}((q,i)) = \frac{x_i}{4\sum_ix_i}$;
        \item ${v_p}((p,p')) =  {v_q}((q,q')) = \frac{3}{4}$,  ${v_{p'}}((p,p')) = {v_q'}((q,q')) = 1$, and ${v_{p'}}((z,p')) = 0$;
        \item $ {v_w}((w,z)) = {v_z}((w,z)) = {v_x}((x,y)) = {v_y}((x,y)) = \frac{3}{4}$, and ${v_a}((a,b)) =  {v_b}((a,b)) = \frac{1}{8}$ for each pair of vertices $(a,b)\in\{(w,x),(x,z), (z,y), (y,w)\}$.  
    \end{itemize}

We first claim that in an EQ1 orientation $\pi$, each agent has a value of at least $\frac{1}{8}$. Let us focus on $w,x,y,z$, and in any orientation (and hence in $\pi$), there must be an agent receiving two edges, with one having a value of $\frac{1}{8}$ and another having a value of $\frac{3}{4}$.
Thus, after removing the edge with value $\frac{3}{4}$, that agent still has value $\frac{1}{8}$. As a result, in an EQ1 orientation $\pi$, each agent should receive a value of at least $\frac{1}{8}$.
Then edges $(p',p)$ and $(q',q)$ must be allocated to $p'$ and $q'$ respectively, and for each $i\in [n]$, edge $(i',i)$ should be allocated to $i'$.
Moreover, for each $i\in [n]$, vertex $i$ should receive at least one (indeed exactly one) of $(i,p)$ and $(i,q)$. Then the total value of the edges that can be allocated to $p$ and $q$ is $\frac{1}{4}$, which makes in $\pi$, the value of $p$ and $q$ should be $\frac{1}{8}$. Thus, there exists a subset of $\{x_1,\ldots,x_n\}$ with the total value $T$.
Therefore, there exists an EQ1 orientation if and only if the \ssc{partition} instance is a yes-instance.

As for the chores-instance, we consider the same graph creation but negate agents' valuations.
Similarly, in any orientation, there must be an agent among $w,x,y,z$ receiving at least one edge with value $-\frac{1}{8}$ and one edge with value $-\frac{3}{4}$.
Thus, in an EQ1 orientation $\pi$, every agent should receive a value of at most $-\frac{1}{8}$. Then edges $(p',p)$ and $(q',q)$ must be allocated to $p'$ and $q'$, respectively. Moreover, for each $i\in [n]$, edge $(i',i)$ should be allocated to $i'$.
Then one can verify that the unallocated edges can make both $p$ and $q$ have a value of no greater than $-\frac{1}{8}$ if and only if the \ssc{partition} instance is a yes-instance.
Therefore, we establish an equivalence between the existence of EQ1 orientation and a yes-instance of the \ssc{partition} problem.
\end{proof}

The non-existence of EQ1 orientations directly implies that EQX orientations are not guaranteed to exist, as EQX is a notion stricter than EQ1, while the computational hardness of the EQ1 decision problem does not directly extend to EQX.
Below, we establish the hardness of the decision problem concerning EQX.

\begin{theorem}\label{thm::EQx}
    Determining whether an EQX orientation exists on a simple graph in both goods- and chores-instances is NP-complete.
\end{theorem}
\begin{proof}
The decision problem is in NP. 
We derive the reduction from the \ssc{partition} problem: whether a set of integers $x_1,\ldots,x_n$ can be divided into two subsets with equal sums.

We first prove for goods and create a goods-instance of a simple graph, as illustrated in Figure \ref{fig:EQXNPC}:
\begin{itemize}
    \item create vertices $1,\ldots,n$ and vertices $p,q$; for each $j\in [n]\cup\{p,q\}$, create two dummy vertices $k_j,k'_j$;
    \item for each $j\in [n]\cup\{p,q\}$, create edges $(j,k_j), (j,k'_j), (k_j,k'_j)$. Moreover, for each $j\in [n]$, create edges $(j,p)$ and $(j,q)$.
\end{itemize}
\begin{figure}[h]
    \centering
    \includegraphics[width=0.6\textwidth]{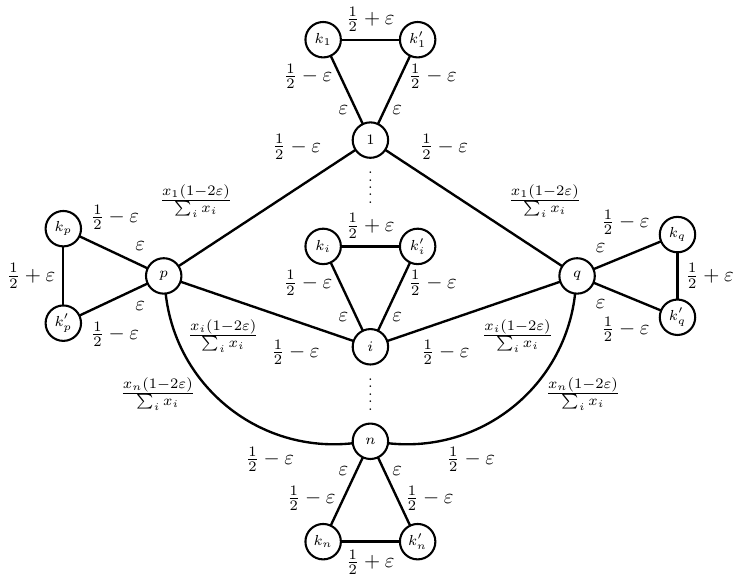}
    \caption{Illustration of the goods-instance for reduction from EQX. The label closer to a vertex represents that vertex’s value for the edge.}
    \label{fig:EQXNPC}
\end{figure}
We define the agents' valuations as follows:
\begin{itemize}
    \item for each $j\in [n]\cup \{p,q\}$, let $ {v_j}((j,k_j)) =  {v_j}((j,k_j')) = \epsilon$, $ {v_{k_j}}((j,k_j)) =  {v_{k_j'}}((j,k_j')) = \frac{1}{2}-\epsilon$, and $ {v_{k_j}}((k_j,k_j')) =  {v_{k_j'}}((k_j,k_j')) = \frac{1}{2} + \epsilon$;
    \item for each $i\in [n]$, let ${v_i}((p,i)) =  {v_i}((q,i)) = \frac{1}{2}-\epsilon$;
    \item for $p$ and $q$, let ${v_p}((p,i)) = v_q((q,i)) = \frac{x_i}{\sum_i x_i}(1 - 2\epsilon)$ for all $i\in [n]$,
\end{itemize}
where $\epsilon>0$ is arbitrarily small. As the total value of each agent is 1, the created instance is normalized.

We claim that in an EQX orientation $\pi$, each $j\in [n] \cup \{p,q\}$ must receive exactly one of the edges $(j,k_j)$ and $(j,k_j')$. Suppose not. If $j$ receives neither $(j,k_j)$ nor $(j,k_j')$, then EQX is violated between $k_j$ and $k'_j$;
assume without loss of generality that $(k_j,k'_j)$ is allocated to $k_j$, then $k'_j$ violates EQX when compared to $k_j$ as 
\[
{v_{k_j}}(\pi_{k_j} \setminus \{(p,k_j)\}) = \frac{1}{2}+\epsilon > {v_{k_j'}}(\pi_{k_j'}) = {v_{k_j'}}((j,k_j')) = \frac{1}{2} - \epsilon.
\]
If $j$ receives both $(j,k_j)$ and $(j,k_j')$, then one of $k_j$ and $k'_j$ receives no edges and violates EQX when compared to $j$.

Given the above claim, we further show that in an EQX orientation $\pi$, each agent $i\in [n]$ must receive exactly one of the edges $(p,i)$ and $(q,i)$. Suppose not.
If $i$ receives both $(p,i)$ and $(q,i)$, then $k_i$ and $k'_i$ violate EQX when compared to $i$, as after removing the edge with value $\epsilon$ for $i$ (such an edge exists due to the above claim), the value of $i$ is still $1-2\epsilon > \frac{1}{2}+\epsilon$.
If $i$ receives neither $(p,i)$ nor $(q,i)$, then $i$ violates EQX when compared to $p$, as after removing the edge with value $\epsilon$ for $p$, the value of $p$ is greater than $2\epsilon$.

One can verify that an EQX orientation exists if and only if the \ssc{Partition} instance is a yes-instance. 
Suppose that the \ssc{partition} instance is a yes-instance, and the solution consists of $I_1$ and $I_2$. For each $i\in[n]$, if $i\in I_1$, allocate $(i,p)$ to $p$ and if $i\in I_2$, allocate $(i,q)$ to $q$. Next, for each $i\in [n]$, allocate $i$ the other incident edge with a value of $\frac{1}{2}-\epsilon$ for her.
For each $j\in [n]\cup\{p,q\}$, allocate $(k_j,j)$ to $j$, $(k'_j,j)$ to $k'_j$, and $(k'_j,k_j)$ to $k_j$.
At this point, all edges are allocated, and the orientation is EQX, as $I_1$ and $I_2$ are a solution to the \ssc{partition} instance.

For the reverse direction, suppose that there exists an EQX orientation $\pi$. Let $S_p:=\{i\in [n]\mid (i,p)\in \pi_p \}$ and $S_q := \{i\in [n] \mid (i,q)\in  \pi_q\}$. We prove that $S_p$ and $S_q$ must form a solution to the \ssc{partition} instance. 
First, as each $i$ receives exactly one from edges $(i,p)$ and $(i,q)$, it holds that $S_p\cup S_q = [n]$.
Assume for the contradiction that $S_p$ and $S_q$ do not form a solution, and then, assume $\sum_{i\in S_p} x_i < \sum_{i \in S_q} x_i$. Moreover, as $x_i$'s are integers, we have $\sum_{i\in S_p} x_i \leq T-1$, and thus, 
\[
v_p(\pi_p) \leq \frac{T-1}{2T}(1-2\epsilon) + \epsilon = \frac{1}{2}+\frac{\epsilon}{T} - \frac{1}{2T} < \frac{1}{2}-\epsilon,
\]
where the last inequality transition is derived from $\epsilon\ll \frac{1}{T}$.
For any $i\in [n]$, as $\pi_i$ contains an edge with value $\epsilon$ and an edge with value $\frac{1}{2}-\epsilon$, agent $p$ violates EQX when compared to $i$, contradicting that $\pi$ is an EQX orientation.

As for the chores-instance, we consider the same graph creation but negate agents' valuations.
We claim that in an EQX orientation $\pi$, each $j\in[n]\cup \{p,q\}$ must receive exactly one of the edges  $(j,k_j)$ and $(j,k'_j)$. If $j$ receives both $(j,k_j)$ and $(j,k'_j)$, then one of $k_j$ and $k'_j$ receives no edges in $\pi$ (let $k_j$ be such a vertex), and $j$ violates EQX when compared to $k_j$.
If $j$ receives neither $(j,k_j)$ nor $(j,k'_j)$, then EQX is violated when comparing $k_j$ and $k'_j$.
Based on the claim, we further show that in an EQX orientation $\pi$, each $i\in [n]$ should receive exactly one of edges $(p,i)$ and $(q,i)$; otherwise, $i$ violates EQX when compared to the vertex with the larger value between $k_i$ and $k'_i$.
Then by the argument similar to that of goods, one can verify that there exists an EQX orientation in the created instance if and only if the \ssc{partition} instance is a yes-instance.
\end{proof}

We remark that all hardness results in this section hold for normalized instances, where the total values of items for agents are identical.

\section{EF1 Orientations for Chores}

In this section, we focus on the EF1 orientation for chores that have not been covered in the literature.
For goods, it is known that an EF1 orientation exists in the general orientation model \citep{DBLP:journals/corr/abs-2409-13616}. 
However, for chores, EF1 orientations may not exist, even for simple graphs.

\begin{theorem} \label{thm::EF1}
    For the chores-instance with simple graphs, EF1 orientations do not always exist. 
    For simple graphs, one can compute an EF1 orientation in polynomial time when such an orientation exists.
\end{theorem}

We say an edge is \emph{objectively negative} if it results in negative values for both endpoints.
We present a necessary and sufficient condition for an instance with simple graphs to admit EF1 orientations.

\begin{proposition} \label{prop:pigeonhole}
    In a chores-instance, there is an EF1 orientation on a simple graph if and only if the number of objectively negative edges is not greater than the number of vertices in the graph.
\end{proposition}
\begin{proof}
We begin with a claim that if there exists an EF1 orientation $\pi$, then each agent should receive at most one objectively negative edge. Suppose not, and assume that $i$ receives two objectively negative edges $(i,j)$ and $(i,k)$ in $\pi$, 
then $i$ will violate EF1 when compared to $j$, as after removing any edge in $\pi$, the value of $i$ is negative, but $i$'s value for $\pi_j$ is zero (as $(i,j)$ is in $\pi_i$, then $\pi_j\cap E_i=\emptyset$); 
recall that $i$ values $e$ at zero for all $e\notin E_i$.
Thus, based on this claim, if there exists an EF1 orientation, the number of objectively negative edges is at most the number of vertices.

For the ``if'' direction, when the number of objectively negative edges is at most the number of vertices, we can compute an EF1 orientation. First, allocate all edges to an incident agent who has zero values on them. 
Then, if there are edges remaining unallocated, they must be objectively negative; they will form one or more components. If a component formed by the objectively negative edges contains no cycle, i.e., is a tree, we can allocate every edge downwards to the corresponding child vertex by selecting an arbitrary vertex as the root of the tree. 
If a component contains a cycle, we can allocate every edge in the cycle to an endpoint in the same direction, and the remaining edges will form a tree. 
Repeat this process until all edges are allocated. The procedures terminate in at most $m$ rounds, as in each round, at least one edge is allocated.
\end{proof}

The characterization above indicates that if, in a simple graph, the number of objectively negative edges exceeds the number of vertices, then no EF1 orientation will exist. Such an instance clearly exists, demonstrating that EF1 orientations do not always exist for chores.
Moreover, the characterization directly leads to a polynomial-time algorithm for deciding the existence of EF1 orientations for simple graphs and for computing one when it exists.

On the contrary, the decision problem becomes computationally intractable for multigraphs.

\begin{theorem}\label{thm::EF1-multi}
For the chores-instance, determining whether an EF1 orientation exists in a multigraph is NP-complete.
\end{theorem}
\begin{proof}
    The decision problem is in NP, as given an orientation, one can decide whether it is EF1 in polynomial time. We  derive below the reduction from the \ssc{partition} problem. Create a chores-instance with 3 vertices and $n+4$ edges as follows:
\begin{itemize}
    \item create vertices $1,2,3$;
    \item create $n$ edges $e_1,\ldots,e_n$ between 1 and 2. Moreover, create two edges $e_{n+1},e_{n+2}$ (resp., $e_{n+3},e_{n+4}$) between 1 and 3 (resp., between 2 and 3);
    \item for each $k\in [4]$, the value of $e_{n+k}$ is $-2T-1$ for both of its endpoints, and for each $i\in [n]$, the value of $e_i$ is $-x_i$ for both of its endpoints.
\end{itemize}

We claim that in an EF1 orientation $\pi$, agent 3 cannot receive both $e_{n+1},e_{n+2}$ or both $e_{n+3},e_{n+4}$. Suppose that 3 receives both $e_{n+1},e_{n+2}$, then agent 3 would value agent 1's bundle at zero and violates EF1. The same reasoning applies to $e_{n+3},e_{n+4}$.
Similarly, agent 1 (resp., agent 2) cannot receive both $e_{n+1},e_{n+2}$ (resp., $e_{n+3},e_{n+4}$).
Therefore, in an EF1 orientation $\pi$, agent 1 (resp., agent 2) must receive exactly one of $e_{n+1},e_{n+2}$ (resp., $e_{n+3},e_{n+4}$). Then at this point, 3 satisfies EF1 no matter how the remaining edges are allocated.

For agents 1 and 2, they do not violate EF1 when compared to agent 3, as the total value of $e_1,\ldots,e_n$ is greater than $-2T-1$. When verifying the EF1 condition between agents 1 and 2, edge $e_{n+k}$ for some $k$ would always be hypothetically removed.
Therefore, there exists an EF1 orientation if and only if the \ssc{partition} instance is a yes-instance.
\end{proof}

\section{Conclusion}

In this paper, we studied proportionality, equitability, and their relaxations in the fair orientation problem and presented a picture of existence and computational results on these fairness notions.
We also addressed the previously unexplored question of EF1 orientations for chores.
We hope these results contribute to a deeper understanding of these widely studied fairness notions in the orientation model.

Looking forward, as the hardness results in Section \ref{sec:EQ} are derived from \ssc{partition}, it remains unknown whether pseudo-polynomial-time algorithms exist.
For proportionality, it is worthwhile to consider PROPm \citep{DBLP:conf/aaai/BaklanovGGS21}, a notion lying between PROP1 and PROPX, and to investigate whether PROPm orientations always exist. The non-existence example for PROPX presented in this paper does not imply the non-existence of PROPm orientations.

From a broader perspective, developing new fairness notions applicable in the general orientation model, with no restrictions on the number of relevant agents per item, is a promising and valuable direction. 
Our work has provided some negative results regarding the existing fairness notions and can serve as a foundation for the development of new fairness concepts.

\section*{Acknowledgments}
This work is supported in part by the Guangdong Basic and Applied Basic Research Foundation (Grant No. 2024A1515011524) and the Hong Kong SAR Research Grants Council (Grant No. PolyU 15224823).

\bibliographystyle{ACM-Reference-Format}
\bibliography{ref}


\newpage
\appendix

\section*{Appendix}


\section{Proof of Theorem \ref{thm:prop1:fPO}}\label{sec::prop:meta}

First, we present a key lemma.

\begin{lemma}\label{lem::fpo-2-ak-last}
	For any orientation $\pi$, one can compute in polynomial time a fPO orientation $\pi^*$ such that: (i) $\pi^*$ either Pareto dominates $\pi$ or gives every agent the same value as $\pi$, and (ii) the undirected consumption graph $\mathcal{CG}_{\pi^*}$ is acyclic.
\end{lemma}
To prove the lemma, we need extra notation. In the following, we say that an item $e$ is:
\begin{itemize}
    \item a chore if $v_i(e)<0$ for all $i \in N_e$;
    \item neutral if $v_i(e)=0$ for at least one $i\in N_e$ and $v_j(e)\leq 0$ for all $j\in N_e$;
    \item a good if $v_i(e)>0$ for at least one $i\in N_e$;
    \item a pure good if $v_i(e)>0$ for all $i\in N_e$.
\end{itemize}

An orientation $\pi$ is \emph{non-malicious} if each good $e$ is consumed by agents $i\in N_e$ with $v_i(e)>0$ and each neutral item $e$ is consumed by agents $i \in N_e$ with $v_i(e)=0$. 
Every fPO orientation is clearly non-malicious.

We will consider agent-object graphs, which are bipartite graphs where the nodes on one side represent agents and the nodes on the other side represent objects.
In the \emph{(undirected) consumption-graph} $\mathcal{CG}_{\pi}$ of an orientation $\pi$, there is an edge between agent $i\in [n]$ and an item $e\in E$ if and only if $\pi_{i,e}>0$; note that since $\pi$ is an orientation, if there is an edge between $i$ and $e$, then $i\in N_e$ and $e\in E_i$.

The \emph{weighted directed consumption-graph} $\overrightarrow{\mathcal{CG}}_{\pi}$ of an orientation $\pi$ is constructed as follows. There is an edge $(i\rightarrow e)$ with weight $w_{i\rightarrow e} =|v_i(e)|$ if one of the two conditions holds:
\begin{itemize}
    \item $\pi_{i,e}>0, v_i(e) \geq 0$ and $e\in E_i$;
    \item $\pi_{i,e}<1, v_i(e)<0$ and $e\in E_i$.
\end{itemize}
The opposite edge $(e\rightarrow i)$ with weight $w_{e\rightarrow i}=\frac{1}{|v_i(e)|}$ is included in $\overrightarrow{\mathcal{CG}}_{\pi}$ in one of the two cases:
\begin{itemize}
    \item $\pi_{i,e}>0, v_i(e)<0$ and $e\in E_i$;
    \item $\pi_{i,e}<1, v_i(e)>0$ and $e\in E_i$.
\end{itemize}

The \emph{product} of a directed path $P$ in $\overrightarrow{\mathcal{CG}}_{\pi}$, denoted $\Pi(P)$, is the product of weights of edges in $P$. In particular, the product of a cycle $C=(i_1\rightarrow e_1 \rightarrow \cdots \rightarrow e_L \rightarrow i_{L+1}=i_1)$ is
$$
\Pi(C) = \prod_{k=1}^L(w_{i_k\rightarrow e_k}\cdot w_{e_k\rightarrow i_{k+1}}).
$$

Closely following the steps in \cite{DBLP:journals/ior/SandomirskiyS22}, we present the characterization of fPO orientation. 
\begin{lemma}\label{lem::fpo-characterisation}
    Given an orientation $\pi$, the following three properties are equivalent 
     \begin{enumerate}[label=(\roman*)]
        \item $\pi$ is fPO;
        \item $\pi$ is non-malicious and its directed consumption graph $\overrightarrow{\mathcal{CG}}_{\pi}$ has no cycle $C$ with $\Pi(C)<1$.
        \item there is a vector of weights $\lambda=(\lambda_i)_{i\in [n]}$ with $\lambda_i>0, i\in [n]$, such that $\pi_{i,e}>0$ implies $\lambda_i v_i(e) \geq \lambda _j v_j(e)$ for all $i,j\in [n]$ and $e\in E$.
     \end{enumerate}
\end{lemma}
\begin{proof}

(i) $\implies$ (ii). 
If $\pi$ is fPO but malicious, then reallocating items in a non-malicious way strictly improves the value of some agents without harming the others.
Thus, an fPO orientation $\pi$ must be non-malicious.

We now show that there are no directed cycles $C=(i_1\rightarrow e_1 \rightarrow i_2 \rightarrow e_2 \rightarrow \ldots \rightarrow i_L \rightarrow e_L \rightarrow i_{L+1}=i_1)$ in $\overrightarrow{\mathcal{CG}}_{\pi}$ with $\Pi(C)<1$. 
For the sake of contradiction, assume that $C$ is such a cycle. We show how to construct an exchange of items among the agents in $C$ such that their value strictly increases without affecting the other agents. This will contradict the Pareto-optimality of $\pi$.

Define $R:=\Pi(C)^{\frac{1}{L}}$. As $\Pi(C)<1$, then $R<1$ holds. For each $k\in [L]$, there are edges $i_k\rightarrow e_k$ and $e_k\rightarrow i_{k+1}$. Then, by the definition of $\overrightarrow{\mathcal{CG}}_{\pi}$,
\begin{itemize}
    \item either $i_k$ receives a positive amount of $e_k$ and both $i_k$ and $i_{k+1}$ have positive values for $e_k$,
    \item or $i_{k+1}$ has a positive amount of $e_{k+1}$ and both $i_k$ and $i_{k+1}$ have negative values for $e_k$.
\end{itemize}
Now we describe the reallocation. For each $k\in [L]$, agent $i_k$ gives a small positive amount $\epsilon_k$ of $e_k$ to $i_{k+1}$ in the case of goods or $i_{k+1}$ gives $\epsilon_k$ fraction of $e_k$ to $i_k$ in the case of chores where $\epsilon_k \in (0,h_k]$ and $h_k=\pi_{i_k,e_k}$ for a good and $h_k=\pi_{i_{k+1},e_k}$ for a chore.
Then, each $i_k$ loses a value of $\epsilon_{k}\cdot |v_{i_k}(e_k)|$ and gains a value of $\epsilon_{k-1}\cdot |v_{i_k}(e_{k-1})|$, resulting in a net change of $\epsilon_{k-1}\cdot |v_{i_k}(e_{k-1})| - \epsilon_{k}\cdot |v_{i_k}(e_k)|$ in the value of agent $i_k$.
In order to guarantee that every agent in $C$ strictly gains from the reallocation, it is sufficient to choose $\epsilon_1,\ldots,\epsilon_k$ such that the following inequalities hold for all $k\in [L]$:
$$
\epsilon_{k-1}\cdot |v_{i_k}(e_{k-1})| - \epsilon_{k}\cdot |v_{i_k}(e_k)| > 0 \Longleftrightarrow \frac{\epsilon_{k-1}}{\epsilon_k} > \frac{|v_{i_k}(e_k)|}{|v_{i_k}(e_{k-1})|}.
$$

For any $\epsilon_1 >0$, define $\epsilon_k=\epsilon_{k-1} \cdot R\cdot \frac{|v_{i_k}(e_{k-1})|}{|v_{i_k}(e_k)|}$ for every $k\in \{2,\ldots,L\}$. Since $R<1$, the above inequality is satisfied for each $k\in \{2,\ldots,L\}$. It remains to show that it is also satisfied for $k=1$.
Note that
\begin{align*}
    \epsilon_L &=\epsilon_1 R^{L-1} \prod\limits_{k=2}^L\frac{|v_{i_k}(e_{k-1})|}{|v_{i_k}(e_k)|} \\& = \epsilon_1 R^{L-1}   \frac{|v_{i_1}(e_1)|}{|v_{i_1}(e_L)|} \prod\limits_{k=1}^L\frac{|v_{i_k}(e_{k-1})|}{|v_{i_k}(e_k)|}\\
    & =\epsilon_1 \frac{R^{L-1}}{\Pi(C)}\frac{|v_{i_1}(e_1)|}{|v_{i_1}(e_L)|} = \epsilon_1R^{-1}\frac{|v_{i_1}(e_1)|}{|v_{i_1}(e_L)|}, 
\end{align*}
which implies
$$
\frac{\epsilon_1}{\epsilon_L} = R\cdot \frac{|v_{i_1}(e_L)|}{|v_{i_1}(e_1)|} < \frac{|v_{i_1}(e_L)|}{|v_{i_1}(e_1)|},
$$
as required. Therefore, we can choose $\epsilon_1$ sufficiently small so that $\epsilon_k \leq h_k$ for all $k\in [L]$ and such a reallocation is feasible.

(ii) $\implies$ (iii).

We assume that $\overrightarrow{\mathcal{CG}}_{\pi}$ contains no directed cycles $C$ with $\Pi(C)<1$ and $\pi$ is non-malicious. We prove below the existence of weights $\lambda_i$'s satisfying (iii).

Starting from $\overrightarrow{\mathcal{CG}}_{\pi}$, for each pair of distinct agents $i,j\in [n]$, add directed edges $i\rightarrow j$ with weight

{\begin{align*}
w_{i\rightarrow j} = \left( 
    \begin{multlined}[t]
        \max_{e\in E, k\in N_e, v_k(e) \neq 0}\left\{ 1, |v_k(e)|, \frac{1}{|v_k(e)|}  \right\}
    \end{multlined} 
\right) ^{2(n-1)}
\end{align*}}

Note that each new edge has the same weight. Let the resulting new graph be $\overrightarrow{G}$, and we claim that $\overrightarrow{G}$ has no cycle $C$ with weight $\pi(C)<1$.
For the sake of contradiction, assume that $C$ is such a cycle in $\overrightarrow{G}$ with $\pi(C)<1$. Then $C$ must contain at least one new edge.
Each agent vertex appears once in $C$, the cycle $C$ contains at most $2n-2$ old edges. If none of the old edges in $C$ has weight zero, then by the definition of the weight of the new edge, $\pi(C)\geq 1$ holds, a contradiction.
Old edges with weight zero are only possible from an agent to an item ( suppose from some agent $i$ to some item $e$), and moreover, $\pi_{i,e}>0$ and $v_i(e)=0$.
However, as $\pi$ is non-malicious, such $e$ has no outgoing edges, and thus, edge $i\rightarrow e$ cannot be a part of any cycle.

Fix an arbitrary agent (suppose agent 1). For every other agent $j\in [n]$, let $P_{1,j}$ be a directed path from 1 to $j$ in $\overrightarrow{G}$, for which the product $\Pi(P_{1,j})$ is minimal.
Note that the minimum is well-defined and is attained on an acyclic path, as the construction ensures that there is no cycle with a product smaller than 1.

Define $\lambda_j := \Pi(P_{1,j})$ for all $j\neq 1$ and $\lambda_1:=1$. We now prove that these weights satisfy (iii), i.e., $\pi_{i,e}>0$ implies $\lambda_i v_i(e) \geq \lambda_j v_j(e)$ for all $j\in N_e$.
Fix $i,e$ with $\pi_{i,e}>0$ and some $j$ with $j\in N_e$. As $\pi$ is non-malicious, we can, without loss of generality, assume that $i$ and $j$ have an agreement on whether $e$ is a good or a chore, i.e., $v_i(e) \cdot v_j(e) > 0$; if $i$ and $j$ disagree, then by the non-maliciousness, $\lambda_i v_i(e) \geq \lambda_j v_j(e)$ holds for any $\lambda_i,\lambda_j > 0$.

If $e$ is a good (i.e., $v_i(e)>0$ and $v_j(e)>0$), then there is an edge $ i\rightarrow e$ (as $\pi_{i,e}>0$, $v_i(e) > 0 $, and $i\in N_e$) and an edge $e\rightarrow j$ (as $\pi_{j,e}<1$, $v_j(e)>0$, and $j\in N_e$).
Consider the optimal path $P_{1,i}$ and the concatenated path $Q_{1,j}=P_{1,i}\rightarrow e\rightarrow j$. As the path $P_{1,j}$ has the minimal product among all paths from 1 to $j$, we have
\begin{align*}
    \Pi(Q_{1,j}) \geq  \Pi(P_{1,j}) & \Longleftrightarrow \Pi(P_{1,i})\cdot \frac{v_i(e)}{v_j(e)} \geq \Pi(P_{1,j}) \\& \Longleftrightarrow \lambda_i v_i(e) \geq \lambda_j v_j(e).
\end{align*}
If $e$ is a chore (i.e., $v_i(e)<0$ and $v_j(e)<0$), then there is an edge $j\rightarrow e$ (as $v_j(e)<0$, $\pi_{j,e}<1$, and $j\in N_e$) and an edge $e\rightarrow i$ (as $\pi_{i,e}>0$, $v_i(e)<0$, and $i \in N_e$).
Define $Q_{1,i}$ as $P_{1,j} \rightarrow e\rightarrow i$. We have
\begin{align*}
\Pi(Q_{1,i}) \geq \Pi(P_{1,i}) & \Longleftrightarrow  \Pi(P_{1,j}) \cdot \frac{|v_j(e)|}{|v_i(e)|} \geq \Pi(P_{1,i}) \\ &\Longleftrightarrow  \lambda_j|v_j(e)|  \geq \lambda_i |v_i(e)| \\& \Longleftrightarrow  \lambda_i v_i(e) \geq \lambda_j v_j(e).
\end{align*}

(iii) $\implies$ (i).

As in $\pi$, each item $e$ is allocated to the agent $i$ with the highest $\lambda_iv_i(e)$ among all agents in $N_e$, then $\pi$ maximizes $\sum_{e\in E}\sum_{i\in N_e} \lambda_i v_i(e) $ over all (fractional) orientations. Since $\lambda_i$'s are positive, $\pi$ is fPO because any Pareto-improvement must increase the weighted sum.
\end{proof}

We are now ready to prove Lemma \ref{lem::fpo-2-ak-last}.

\begin{proof}[Proof of Lemma \ref{lem::fpo-2-ak-last}]
    If $\pi$ is malicious, implement the following reallocation:
    \begin{itemize}
        \item for each $e\in E$ with $\max_{i\in N_e}v_i(e)>0$, reallocate the share of agents with $v_j(e)\leq 0$ to an agent $i\in N_e$ with $v_i(e)>0$;
        \item for each $e\in E$ with $\max_{i\in N_e}v_i(e)=0$, reallocate the share of agents with $v_j(e)< 0$ to an agent $i$ with $v_i(e)=0$.
    \end{itemize}
    Let the resulting non-malicious orientation be $\pi'$.

    We now describe reallocation of items that eliminates the cycle (if any). 
    Let us call a cycle $C=(i_1\rightarrow e_1 \rightarrow i_2 \rightarrow e_2 \rightarrow \ldots \rightarrow i_L \rightarrow e_L \rightarrow i_{L+1}=i_1)$ in the directed graph $\overrightarrow{\mathcal{CG}}_{\pi'}$ a \emph{simple} graph if each node is visited at most once and for any $i\in [n]$ and $e\in E_i$, only one of edges $i\rightarrow e$ or $e\rightarrow i$ is contained in the cycle.

    We claim that if there is a simple cycle $C$ in $\overrightarrow{\mathcal{CG}}_{\pi'}$ with $\Pi(C)\leq 1$, then $C$ can be eliminated by reallocation of items. The idea of eliminating the cycle is similar to the reallocation in the proof of Lemma \ref{lem::fpo-characterisation}.
    As edges $i_k\rightarrow e_k$ and $e_k\rightarrow i_{k+1}$ exist in $\overrightarrow{\mathcal{CG}}_{\pi'}$, the values $v_{i_k}(e_k)$ and $v_{i_{k+1}}(e_k)$ are both non-zero and have the same sign due to the construction of $\overrightarrow{\mathcal{CG}}_{\pi'}$.
    We implement the following reallocation of items:
    \begin{itemize}
        \item if $v_{i_k}(e_k)>0$ and $v_{i_{k+1}}(e_k)>0$, then take $\epsilon_k$ amount of $e_k$ from $i_k$ and give it to $i_{k+1}$, where $0< \epsilon \leq h_k$ with $h_k=\pi'_{i_k,e_k}$;
        \item if $v_{i_k}(e_k)<0$ and $v_{i_{k+1}}(e_k)<0$, then transfer $\epsilon_k$ of $e_k$ from $i_{k+1}$ to $i_k$ where $0< \epsilon \leq h_k$ with $h_k=\pi'_{i_{k+1},e_k}$.
    \end{itemize}
    Then amounts $\epsilon_k$'s are selected in a way such that $\epsilon_k|v_{i_k}(e_k)| = \epsilon_{k+1}|v_{i_k}(e_{k+1})|$ for every $k\in [L-1]$. Thus for each $k=2,\ldots,L$, the value of agent $i_k$ remains indifferent after the reallocation, while agent $i_1$ is weakly better off because the condition of $\Pi(C) \leq 1$.
    Note that after reallocating items, the resulting allocation is still an orientation. It is not hard to verify that we can select $\epsilon_k$'s as large as possible so that one of edges $i_k\rightarrow e_k$ in $\overrightarrow{\mathcal{CG}}_{\pi'}$ can be removed.

    Repeat this reallocation until there are no simple cycles with $\Pi(C) \leq 1$. As at least one edge is removed each time in the undirected graph of the underlying orientation and there are at most $mn$ edges, we need up to $(n-1)m$ repetitions for achieving the orientation where no simple cycle has product at most 1.
    Let the resulting orientation be $\pi^*$.

    By the construction, $\pi^*$ weakly improves the value of each agent in $\pi$; is non-malicious; has no cycle $C$ in $\overrightarrow{\mathcal{CG}}_{\pi^*}$ with $\Pi(C)<1$. By Lemma \ref{lem::fpo-characterisation}, $\pi^*$ is fPO.

    We now prove that $\mathcal{CG}_{\pi^*}$ is acyclic. For a contradiction, assume that there is a cycle $C$ in $\mathcal{CG}_{\pi^*}$.
    Then in the directed graph $\overrightarrow{\mathcal{CG}}_{\pi^*}$, there are two cycles: $C$ passed in one direction and in the opposite direction. Let them be $\overrightarrow{C}$ and $\overleftarrow{C}$. As $\Pi(\overrightarrow{C})\cdot \Pi(\overleftarrow{C}) = 1$, one of them has the product of at most 1, and indeed, by fPO it must hold that $\Pi(\overrightarrow{C})=\Pi(\overleftarrow{C}) = 1$.
    However, all such cycles were eliminated in the previous stages.

    For the running time of the algorithm, constructing the non-malicious orientation, finding the cycles with $\Pi(C)\leq 1$, and resolving/eliminating the cycle takes time polynomial in $m$ and $n$. As cycle-elimination repeats at most $(n-1)m$ times, the total running time is polynomial in $n$ and $m$.
\end{proof}

We are now ready to prove Theorem \ref{thm:prop1:fPO}.

\noindent
\begin{proof}[Proof of Theorem \ref{thm:prop1:fPO}]
	Let $\pi$ be the orientation where $\pi_{i,e}=\frac{1}{n_e}$ for all $i\in [n]$ and all $e\in E_i$. One can verify that in $\pi$, each agent $i$ receives exactly her proportional share. 
	Based on Lemma \ref{lem::fpo-2-ak-last}, we find another orientation $\pi^*$ where every agent is weakly better off, and hence, each agent $i$ receives at least her proportional share.
	Moreover, $\pi^*$ is fPO and the undirected consumption graph $\mathcal{CG}_{\pi^*}$ is acyclic.
	
We round $\pi^*$ to an integral orientation that is fPO and PROP1. Consider an \emph{undirected shared graph} $\mathcal{SG}_{\pi^*}$ defined as follows: create a vertex for each agent, and there is an edge between vertices $i$ and $j$ if and only if there exists $e$ such that $ \pi^*_{i,e}, \pi^*_{j,e}>0$.
	In other words, the existence of an edge between $i$ and $j$ in $\mathcal{SG}_{\pi^*}$ means agents $i$ and $j$ share an item in $\pi^*$.
	Moreover, as $\pi^*$ is fPO, each shared item must result in the value with the same sign (positive, negative, or zero) for both its endpoints; otherwise, we can increase the value of an agent without harming others, contradicting fPO.
	
	As $\mathcal{CG}_{\pi^*}$ is acyclic, it is not hard to verify that $\mathcal{SG}_{\pi^*}$ is also acyclic, and hence, $\mathcal{SG}_{\pi^*}$ is a forest. 
	For each component of $\mathcal{SG}_{\pi^*}$, we form it into a rooted tree. Then starting from the leaves to the root, we visit every vertex. For each vertex $i$, we let agent $i$ receive fully the item with non-negative value shared with the agent having a higher depth than that of agent $i$, and let agent $i$ receive fully the item with negative value shared with the agent having lower depth than that of agent $i$.
	
	Let $\pi'$ be the resulting integral orientation. We claim that $\pi'$ is PROP1, as compared to $\pi^*$, the value of each agent $i$ for non-negative item weakly increases and there exists at most one item with negative value that is partially allocated to $i$ in $\pi^*$ but is fully allocated to $i$ in $\pi'$.
	After removing such a negative-valued item,
    agent $i$ has value at least her proportional share, and therefore, $\pi'$ is PROP1.
	
	Finally, we prove that $\pi'$ is fPO. Varian \cite{VARIAN1976249} proved that an allocation $\pi$ is fPO if and only if there exist strictly positive weights $\{\lambda_i\}_{i\in [n]}$ of agents such that $\pi$ maxmizes weighted welfare \\ $\sum_{i\in [n]} \lambda_iv_i(\pi_i)$. We show that this characterization holds for the general orientation model and the mixed-instance in the appendix. 
	Since orientation $\pi^*$ is fPO, $\pi^*$ maxmizes weighted welfare $\sum_{i\in [n]} \lambda^*_iv_i(\pi^*_i)$ for some strictly positive weight $\{\lambda^*_i\}_{i\in [n]}$ of agents.
	Hence in $\pi^*$, each item is always allocated to the agent with the highest weighted value with respect to $\{\lambda^*_i\}_{i\in [n]}$.
	After rounding $\pi^*$ to $\pi'$, it is still the case that every item is allocated to the agent with the highest weighted value, and thus, $\pi'$ maximizes the weighted welfare for the weights $\{\lambda^*_i\}_{i\in [n]}$ and is fPO. \end{proof}

\end{document}